\documentclass[final,nomarks]{dmtcs-episciences}

\usepackage[utf8]{inputenc}
\usepackage{subfigure}
\usepackage[normalem]{ulem}
\newtheorem{theorem}{Theorem}
\newtheorem{remark}{Remark}
\newtheorem{definition}{Definition}
\usepackage{pstricks}
\usepackage{pst-node}

\usepackage[T1]{fontenc}
\newtheorem{conjecture}{Conjecture}
\newtheorem{claim}{Claim}
\newenvironment{inproof}[1]{\noindent {\bf Proof of Claim #1.}}{\hfill$\blacksquare$ \medskip}
\newcommand{\occ}{\operatorname{occ}}
\usepackage{amsmath}
\usepackage[round]{natbib}
\usepackage{amssymb,amsmath}
\newcommand{\bx}{{\bar x}}
\newcommand{\by}{{\bar y}}

\begin{document}
\bibliographystyle{plainnat}

\title[Tight Euler tours in hypergraphs]{Tight Euler tours in uniform hypergraphs -- computational aspects}

\author{Zbigniew Lonc \and Paweł Naroski \and Paweł Rzążewski}
\affiliation{Faculty of Mathematics and Information Science, Warsaw University of Technology, Poland}

\received{2017-7-2}
\accepted{2017-8-12}
\publicationdetails{19}{2017}{3}{2}{3755}
\maketitle


\begin{abstract}
By a \emph{tight tour} in a $k$-uniform hypergraph $H$ we mean any sequence of its vertices $(w_0,w_1,\ldots,w_{s-1})$ such that for all $i=0,\ldots,s-1$ the set $e_i=\{w_i,w_{i+1}\ldots,w_{i+k-1}\}$ is an edge of $H$ (where operations on indices are computed modulo $s$) and the sets $e_i$ for $i=0,\ldots,s-1$ are pairwise different. A tight tour in $H$ is a \textit{tight Euler tour} if it contains all edges of $H$. We prove that the problem of deciding if a given $3$-uniform hypergraph has a tight Euler tour is NP-complete (even if the maximum codegree of a pair of vertices is bounded by $4$), and that it cannot be solved in time $2^{o(m)}$ (where $m$ is the number of edges in the input hypergraph), unless the ETH fails. We also present an exact exponential algorithm for the problem, whose time complexity matches this lower bound, and the space complexity is polynomial.
In fact, this algorithm solves a more general problem of computing the number of tight Euler tours in a given uniform hypergraph.
\end{abstract}

Problems of existence and construction of Euler tours and Hamilton cycles in graphs are among the most intensively studied problems in graph theory. Recently much attention has also been received by Hamiltonian problems for uniform hypergraphs. They have been studied from both theoretical (see e.g. \cite{RRS1,RRS2}, \cite{KO}, \cite{KK}, \cite{HS}, \cite{KMO}), and computational (see \cite{DHK}, \cite{S}, \cite{KRS}) perspectives and led to such spectacular results as a Dirac-type theorem on tight Hamilton cycles in $3$-uniform hypergraphs (see \cite{RRS2}). By contrast, not much research has been done on problems of existence and construction of Euler tours in uniform hypergraphs (see \cite{LN,LN1} and \cite{BS}).

There are many possible generalizations of the graph-theoretic concepts of a cycle and a tour to uniform hypergraphs. The results on Euler tours mentioned in the preceding paragraph concern so-called Berge cycles and tours. In this paper we deal with {\it tight tours}.

It is worth to mention that the problem of existence of a tight Euler tour is nontrivial even for complete $k$-uniform hypergraphs with $k\geq 3$. An anticipated solution to this problem is a content of the following conjecture originally formulated in slightly different terminology of so-called universal cycles.
\begin{conjecture}
(\cite{CDG}) For every $k\geqslant2$ there exists $n_0(k)$ such that for every $n>n_0(k)$ the complete $k$-uniform hypergraph on $n$ vertices has a tight Euler tour if and only if $\binom{n-1}{k-1}\equiv0\pmod k$.
\end{conjecture}
The conjecture obviously holds for $k=2$. \cite{J,J1} and \cite{H} proved it for $k = 3, 4, 5$ and for $k = 6$ when $n$ is relatively prime with $6$. In the remaining cases the conjecture is open.

Another well-known problem related to the subject of this paper is a conjecture by Baranyai and Katona (see \cite{K}, Conjecture 4.1). Let $\cal F$ be the $k$-uniform hypergraph on the vertex set $\{ 1,2,\ldots,n\}$ whose edges are  
\[
\{ 1,2,\ldots,k\}, \{ k+1,k+2,\ldots,2k\},\ldots,\{ (a-1)k+1,(a-1)k+2,\ldots,ak\},
 \]
where $a=\frac{m}{k}$ and $m$ is the least common multiple of $n$ and $k$ (the elements of the sets are computed modulo $n$). 
The hypergraphs obtained from $\cal F$ by permuting the elements of $\{ 1,2,\ldots,n\}$ are called {\it wreaths}. 
Observe that if $n$ and $k$ are coprime, then a wreath is exactly a tight Hamiltonian cycle (i.e. a tight tour which contains every vertex of a hypergraph exactly once). Baranyai and Katona conjectured that the set of edges of the complete $k$-uniform hypergraph on $n$ vertices can be partitioned into disjoint wreaths. It is worth mentioning that when $k$ divides $n$, then this conjecture reduces to the famous Baranyai theorem (see \cite{B}) on a factorization of the complete uniform hypergraph.

Here are some other problems that motivate studying computational aspects of Eulerian problems in hypergraphs.

In computer graphics, 3-dimensional objects are usually represented by 3-dimensional, usually triangular, meshes (see \cite{PKK}). Storing and processing such a representation as raw data might be quite costly.
Indeed, a trivial way of encoding is to store every triangle separately, i.e.  the coordinates of all three vertices are stored. It requires storing $3m$ vertices (where $m$ is the number of triangles in the considered triangulation).
One of typical methods to deal with this problem is to partition the mesh into long strips of triangles and then encode these strips. Observe that if we manage to arrange the vertices of the mesh in a (cyclic) sequence in such a way that the vertices of every triangle appear exactly once as the three consecutive terms of this sequence and every three consecutive terms of it constitute some triangle, then we will have to store the coordinates of $m$ vertices only. Obviously, it is not always possible (see \cite{AHMS}).

Another motivation originates from database theory. Let $V$ be a set of records of some database and let $E$ be a family of some $k$-element subsets of $V$, each of them is supposed to be an answer to some query (e.g. it is a set of $k$ records which satisfy some given conditions in ``the best'' way). We want to minimize the time needed to answer all possible queries. One way of achieving this goal is to organize the records of the data base as a cyclic list $C$ in such a way that every element of $E$ is represented in $C$ as $k$ consecutive elements. The list $C$ is optimal in this approach if every member of $E$ is represented in $C$ exactly once and every $k$ consecutive elements of $C$ constitute some member of $E$ (see \cite{LTT} for some more details).

In both problems described above we have some $k$-uniform hypergraph (in the former example $k=3$) and we want to find a tight Euler tour in this hypergraph.

In this paper we discuss computational complexity issues related to tight Euler tours. Our main result (Theorem \ref{wykreconetorusy}) says that the problem of deciding if a given $3$-uniform hypergraph has a tight Euler tour is NP-complete .

The construction in the proof of Theorem \ref{wykreconetorusy} gives also some insight into the complexity of algorithms solving the problem.
The complexity assumption that we use here is the Exponential Time Hypothesis (usually referred to as the ETH) (\cite{IP-ETH}). A consequence of this hypothesis, typically used to give conditional lower bounds for the running time of algorithms, is the fact that {\sc 3-Sat} with $t$ variables and $p$ clauses cannot be solved in time $2^{o(t+p)}$~(\cite{IPZ-Sparsification}). We show that the existence of a tight Euler tour in a 3-uniform hypergraph cannot be decided in time $2^{o(m)}$ (where $m$ denotes the number of edges of the input hypergraph), unless the ETH fails. This lower bound is matched with a simple exact algorithm with time complexity $2^m \cdot m^{O(1)}$ and polynomial space complexity (see Section \ref{algo}). In fact, the algorithm solves a more general problem of counting the number of tight Euler trails and tours in the input hypergraph.

\section{Preliminaries}

For a positive integer $n$, we define $[n] := \{1,2,\ldots,n\}$.
By a $k$-\emph{uniform hypergraph} we mean a pair $H=(V,E)$, where $V$ is a non-empty finite set and $E\subseteq\binom{V}{k}:=\{X\subseteq V:|X|=k\}$. For a hypergraph $H=(V,E)$, by $E(H)$ we denote the set $E$. The elements of $V$ are called \emph{vertices} of $H$ and the elements of $E$ its \emph{edges}. The codegree of a pair of different vertices $vu$ in a hypergraph is the number of edges in $E$ containing both $v$ and $u$.

A \emph{tight walk} in a $k$-uniform hypergraph $H$ is a sequence  $w_1w_2\ldots w_{s}$ of vertices of $H$ such that for all $i=1,\ldots,s-k+1$ the set $e_i := \{w_i,w_{i+1},\ldots,w_{i+k-1}\}$ is an edge of $H$.

For a tight walk $C=w_1w_2\ldots w_{s}$, we say that an edge $e$ is \emph{contained} in $C$ (or just that $e$ is an edge of $C$) if $e = e_i$ for some $i=1,\ldots,s-k+1$. Note that a single edge may appear in a tight walk more than once, i.e., it is possible that $e_i = e_j$ for some $i \neq j$. Then, we say that an edge is contained in $C$ several times. If every edge is contained in $C$ at most once, then we call $C$ a \emph{tight trail}.

By the \emph{length} of a tight walk $C=w_1w_2\ldots w_{s}$ we mean the number $s-k+1$.
Note that if $C$ is a tight trail, then the length of $C$ is equal to the number of edges contained in $C$.

A tight trail $C$ is a \emph{tight Euler trail} in $H$ if every edge of $H$ is contained in  $C$ exactly once. Note that the length of a tight Euler trail is equal to $|E(H)|$.

A \emph{tight tour} in $H$ is a closed tight trail, i.e., a tight trail $C=w_1w_2\ldots w_mw_1w_2\ldots w_{k-1}$. Such a tight tour will be denoted by $(w_1,w_2,\ldots,w_m)$. A {\em tight Euler tour} is a tight tour which is a tight Euler trail. Note that the length $m$ of a tight Euler tour in $H$ is equal to $|E(H)|$.

\begin{remark}
\label{symetrie}
Observe that if $(w_1,w_2,\ldots,w_m)$  is a tight tour (respectively, tight Euler tour) in $H$, then the reversed sequence $(w_{m},\ldots,w_2,w_1)$ and each  sequence $(w_i,\ldots,w_{m},w_1,\ldots,w_{i-1})$, for $i=1,\ldots,m$, are also tight tours (resp., tight Euler tours) in $H$. We shall identify all these tours with $C$.
Thus, a tight Euler tour corresponds to exactly $2|E(H)|$ closed tight Euler trails.
\end{remark}

Now we turn to $3$-uniform hypergraphs. Our aim is to prove that the problem of deciding if a given $3$-uniform hypergraph has a tight Euler tour is NP-complete. To simplify the notation we shall write  $xyz$ rather than $\{x,y,z\}$ if $\{x,y,z\}$ is an edge of a $3$-uniform hypergraph.
\label{trzy}
If $e$ is an edge of some $3$-uniform hypergraph $H$ and $C$ is a tight tour in $H$ which contains $e$, then the three vertices of $e$ are some three consecutive terms of $C$.

We shall construct now a hypergraph $H_\ell$ (where $\ell>4$ is an integer) which will be used later to prove our main result. Informally speaking, $H_\ell$ is a $3$-uniform hypergraph whose edges are the triangles in a triangulation of the torus depicted in Figure \ref{najpiekniejszy}.
Here is a formal definition of $H_\ell$.

\begin{definition}
For each integer $\ell>4$ we define a hypergraph $H_\ell=(V,E)$ such that
\[V=\{0,\ldots,\ell+1\}\cup\bigcup_{i=1}^{\ell-1}\{v_i^a,v_i^b\}\]
\noindent and
\[E=\bigcup\limits_{i=1}^\ell E_i,\]
\noindent where
\[E_1=\{10\ell,1v_1^a\ell,v_1^a\ell(\ell+1),v_1^av_1^b(\ell+1),v_1^b(\ell+1)0,v_1^b10\},\]
\[E_i=\{i(i-1)v_{i-1}^a,iv_i^av_{i-1}^a,v_i^av_{i-1}^av_{i-1}^b,v_i^av_i^bv_{i-1}^b,v_i^bv_{i-1}^b(i-1),v_i^bi(i-1)\}\]
\noindent for $i=2,\ldots,\ell-1$, and
\[E_\ell=\{\ell(\ell-1)v_{\ell-1}^a,\ell(\ell+1)v_{\ell-1}^a,(\ell+1)v_{\ell-1}^av_{\ell-1}^b,(\ell+1)0v_{\ell-1}^b,0v_{\ell-1}^b(\ell-1),0\ell(\ell-1)\}.\]
\end{definition}

\vspace{5mm}

\begin{figure}[bht!]
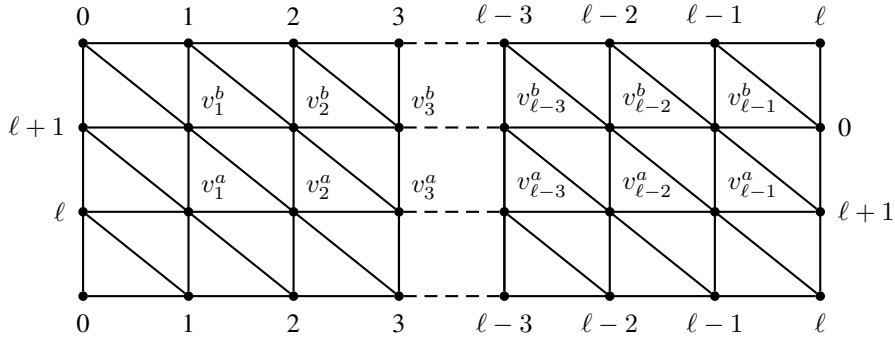

\label{najpiekniejszy}
\begin{center}

\vspace{2mm}

\input najpiekniejszy_ell.tex
\caption{The hypergraph $H_\ell$. In this picture we identify vertices labelled with the same label.}
\end{center}
\end{figure}

Observe that the maximum codegree of a pair of vertices in $H$ is $2$.
We shall study some properties of tight tours in the hypergraph $H_\ell$ now.

Let $e=xyz$ be an arbitrary edge of a tight tour $C=(w_0,\ldots,w_{m-1})$ in $H_\ell$. Then, the vertices of $e$ are some three consecutive terms $w_i,w_{i+1},w_{i+2}$ of $C$. Let us assume first that they appear in the ordering $xyz$, i.e. $w_i=x$, $w_{i+1}=y$, and $w_{i+2}=z$, where the operations on indices are computed modulo $m$. By the definition of a tight tour, $w_{i+1}w_{i+2}w_{i+3}=yzw_{i+3}$ is an edge of $H_\ell$ different from $e$. However, as a pair of distinct vertices of $H_\ell$ is contained either in exactly two edges or in none of the edges, we conclude that the vertex $w_{i+3}$ is uniquely determined by the ordering $xyz$ of the vertices of the edge $e$. The same applies to all other edges contained in $C$, so the ordering of the vertices of some edge contained in $C$ determines the whole tour. Observe that the reversed ordering $zyx$ of the vertices of $e$ defines the same tour $C$. Hence, each edge in $H_\ell$ is contained in at most $3!/2 =3$ tours. From now on we shall identify the mutually reversed orderings of vertices of the edges of $H_\ell$.

Let us describe the tight tours in $H_\ell$ determined by the 3 orderings of the edge $01\ell$ (it does not really matter which edge we choose).

At the beginning let us consider the ordering $01\ell$. It determines the tour $(0,1,\ell,v_1^a,\ell+1,v_1^b)$. This tour contains only $6$ of the $6\ell$ edges of the hypergraph $H_\ell$, so it is not a tight Euler tour. We call such tours in $H_\ell$ the \emph{tours of type $N$}.

Now let us consider the ordering $0\ell1$ of our edge. It determines the tour

\[
(0,\ell,1,v_1^a,2,v_2^a,\ldots,\ell-1,v_{\ell-1}^a,\ell,\ell+1,v_1^a,v_1^b,\ldots,
\]
\[
v_{\ell-1}^a,v_{\ell-1}^b,\ell+1,0,v_1^b,1,\ldots,v_{\ell-1}^b,\ell-1),
\]

\noindent which is a tight Euler tour in $H_\ell$. We call this tight Euler tour in $H_\ell$ the \emph{tour of type} $T$.

Finally, the ordering $\ell01$ determines the tight  tour beginning with
\[(\ell,0,1,v_1^b,v_2^b,v_2^a,v_3^a,3,4,\ldots).\]
It is routine to verify that this tour is an Euler tour in $H_\ell$ if and only if $\ell\not\equiv1\text{ }(\text{mod } 3)$. We call this tight Euler tour in $H_\ell$ the \emph{tour of type} $\mathit{F}$.

Figure \ref{przykladfalszu1} depicts the tight Euler tour of type $F$ in the hypergraph $H_5$.

\vspace{5mm}

\begin{figure}[bht]
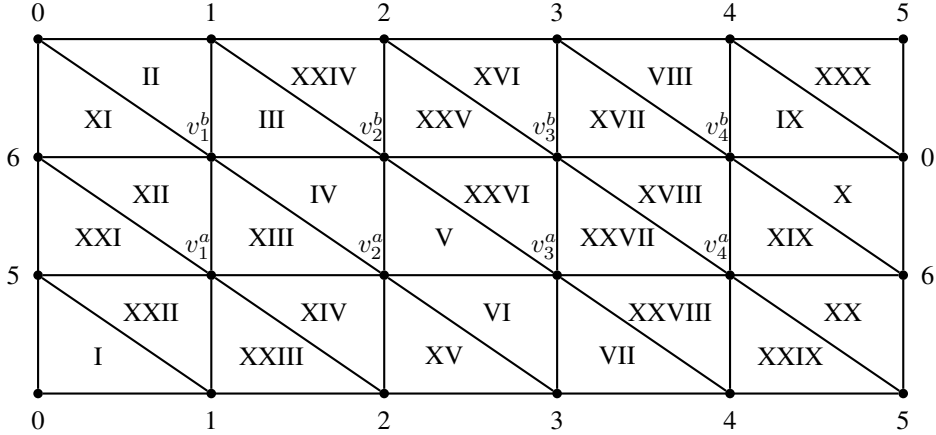

\begin{center}
\input przykladfalszu2.tex
\caption{The tight tour $(5,0,1,v_1^b,v_2^b,v_2^a,v_3^a,3,4,v_4^b,0,6,v_1^b,v_1^a,v_2^a,2,3,v_3^b,$ $v_4^b,v_4^a,6,5,v_1^a,1,2,v_2^b,v_3^b,v_3^a,v_4^a,4)$ of type $F$ in $H_5$. The consecutive edges of this tour are enumerated with Roman numerals.}\label{przykladfalszu1}
\end{center}
\end{figure}

It is clear that there are no more tight Euler tours in $H_\ell$.
From now on we consider only those values of $\ell$, for which the hypergraph $H_\ell$ has exactly two tight Euler tours.

For any edge $e\in E$, the vertices of $e$ are ordered differently in each of the three types of tight tours in $H_\ell$. Each of the two mutually reversed orderings of vertices of $e$ in the tour of type $N$ (resp. of type $T$, of type $F$) will be called an \emph{ordering of $e$ of type} $N$ (resp.  \emph{of type} $T$, \emph{of type} $F$).

Let us observe that the type of an ordering of an edge $e$ is determined by the pair of its vertices which are not consecutive in the ordering. For example, for the edge $01\ell$, the pair $0\ell$ determines an ordering of type $N$, the pair $01$ -- an ordering of type $T$ and the pair $1\ell$ -- an ordering of type $F$. In general, each  pair of vertices joined with a vertical (resp. horizontal, skew) line in Figure \ref{najpiekniejszy} determines orderings of type $N$ (resp. $T$, $F$) of the edges containing the pair.

Let $n>2$. By a \emph{cycle} $C_n$ we mean a $3$-uniform hypergraph

\[(\{0,\ldots,n-1\},\{i(i+1)(i+2):i=0,\ldots,n-1\}),\]

where additions are computed modulo $n$. Note that the maximum codegree of a pair of vertices in $C_n$ is 2.

Clearly, the hypergraph $C_n$ has exactly one tight Euler tour. Abusing the definitions a little bit, we shall  identify sometimes cycles with their tight Euler tours.


\section{Complexity results}

In this section we prove NP-completeness of the problem of determining the existence of a tight Euler tour in a $3$-uniform hypergraph. To make the proof of this statement easier to follow we give a short outline of the main steps of the reasoning.

In the proof we apply a reduction from \textsc{3-Sat}. To this end, for any \textsc{3-Sat} formula, we build a $3$-uniform hypergraph $H$. For every variable $x_i$ appearing in this formula we define a copy of the hypergraph $H_{\ell_i}$ (for some $\ell_i$). As we have seen, the hypergraph $H_{\ell_i}$ (for suitable values of $\ell_i$) has exactly two tight Euler tours. One of them (of type $T$) corresponds to setting $x_i$ as true. The other one (of type $F$) corresponds to setting $x_i$ as false. For every clause $d_j$ of the formula we define a copy of the cycle $C_6$. Some vertices of it are identified with some vertices of the copies of $H_{\ell_i}$ corresponding to the variables appearing in $d_j$. These identifications are designed in such a way that a tour in $H$ can get from the copy of $C_6$ to $H_{\ell_i}$ forcing the tour of type $T$ in $H_{\ell_i}$ if and only if $x_i$ appears positively in $d_j$ and, analogously, such a tour can get from $C_6$ to $H_{\ell_i}$ forcing the tour of type $F$ in $H_{\ell_i}$ if and only if $x_i$ appears negated in $d_j$. The most important property of the hypergraph $H$ is that in every tight Euler tour in $H$ the orderings of all edges of each subhypergraph $H_{\ell_i}$ are of the same type -- either $T$ or $F$. Therefore such an Euler tour defines a truth assignment of the variables of the \textsc{3-Sat} formula (and vice versa).

\begin{theorem}
\label{wykreconetorusy}

The problem of deciding if a given $3$-uniform hypergraph $H$ has a tight Euler tour is NP-complete, even if the maximum codegree of a pair of vertices in $H$ is bounded by $4$.
Moreover, it cannot be solved in time $2^{o(|E(H)|)}$, unless the ETH fails.
\end{theorem}

\begin{proof}
Verifying if a sequence of vertices is a tight Euler tour in a hypergraph can obviously be done in polynomial time, so our problem is in the class NP.

To prove it is complete in this class, we show a reduction from \textsc{3-Sat}.
Let $c=d_1\land\ldots\land d_p$ be a \textsc{3-Sat} formula, where for all $j=1,\ldots,p$, each  $d_j$ is a clause with exactly three literals. We can assume that no clause contains both a variable and its negation because such a clause can be deleted from the formula $c$ without changing its satisfiability. Let $X=\{x_1,\ldots,x_t\}$ be the set of variables appearing in the formula $c$. For a variable $x_i$, let $\occ(x_i)$ denote the number of occurrences of $x_i$ in $c$ (either positive of negative). Moreover, we keep these occurrences ordered. The ordering is implied by the ordering $d_1,d_2,\ldots,d_p$ of the clauses.

For the formula $c$ we construct a $3$-uniform hypergraph $H$.
For every $i = 1,\ldots,t$, let $\ell_i$ be the smallest integer such that $\ell_i \geqslant 2 \occ(x_i) +3$ and $\ell_i \not\equiv1 \pmod{3}$. Obviously, $\ell_i\leqslant2\occ(x_i)+4$.

For each variable $x_i$ we define a copy of the hypergraph $H_{\ell_i}$,  which we denote by $H^i$.
Next, for each clause $d_j$ we define a copy of a cycle $C_6$, which we denote by $C_6^j$.
Finally, we define one copy of the cycle $C_t$.

At the beginning we assume that the vertex sets of all hypergraphs $H^1,\ldots,H^t$, $C_6^1,\ldots,C_6^p$ and $C_t$ are pairwise disjoint. However, we shall identify some vertices in the next steps of the construction of the hypergraph $H$.
In each identification step we will choose two vertices $x,y$, belonging to a common edge, and two different vertices $x',y'$, also belonging to some edge, and will identify $x$ with $x'$ and $y$ with $y'$.
All such pairs $xy$ of identified vertices will be called the \emph{connectors}.

Firstly, for every $i=1,\ldots,t$, we identify the vertices $2i-2$ and $2i-1$ of the hypergraph $C_t$ with the vertices $v_{\ell_i-1}^a$ and $v_{\ell_i-1}^b$, respectively, of the hypergraph $H^i$.

Let $d_j$ be a clause in the formula $c$ and let $x_i$, $x_k$, and $x_s$ be the variables appearing in $d_j$.
Moreover, for $q \in \{i,k,s\}$, let $o_q$ be the number, such that the occurrence of $x_q$ in $d_j$ is the $o_q$-th occurrence of this variable.  If the variable $x_i$ is negated in $d_j$, then we identify the vertices $0$ and $1$ of the hypergraph $C_6^j$ with the vertices $v_{2o_i}^a$ and $v_{2o_i+1}^a$, respectively, of the hypergraph $H^i$. And if a variable $x_i$ is not negated in $d_j$, then we identify the vertices $0$ and $1$ of the hypergraph $C_6^j$ with the vertices $v_{2o_i}^b$ and $v_{2o_i+1}^a$, respectively, of the hypergraph $H^i$. Analogously we identify the vertices $2$ and $3$ of the hypergraph $C_6^j$ with the vertices of $H^k$ and the vertices $4$ and $5$ of $C_6^j$ with the vertices of $H^s$.

These identifications complete the definition of the hypergraph $H$. Note that the hypergraph $H$ can be constructed in a polynomial time for any \textsc{3-Sat} formula $c$.

Observe that in each identification step we identified a pair of vertices from some hypergraph $H^i$ with a pair of vertices from some cycle (either $C_t$ or $C_6^j$). Thus the subhypergraphs of $H$ induced by the edges of the hypergraphs $C_6^1,\ldots,C_6^p$, and $C_t$ are still vertex disjoint. The same is true for the subhypergraphs induced by the edges of $H^1,\ldots,H^t$. Moreover, the maximum codegree of a pair of vertices in $H$ is 4.


\medskip
We shall prove now several claims concerning properties of the hypergraph $H$.

\begin{claim}
Let $C$ be a tight Euler tour in $H$ and let $e$ and $f$ be two edges of some subhypergraph $H^i$ such that they intersect on a $2$-element set, which is not a connector. Then the orderings of $e$ and $f$ defined by the tour $C$ are of the same type in $H^i$.
\end{claim}

\begin{inproof}{1}
Let $e=xyz$ and $f=xyu$. If the ordering of the edge $e$ in $C$ is $zxy$ or $zyx$ (or the reverse of one of them), then, as $xy$ is not a connector, the edge $f$ precedes or succeeds $e$ in the tour $C$. On the other hand, obviously, $e$ and $f$ are consecutive edges of some tour in $H^i$. Therefore, by the definition of the type of an ordering of an edge, the orderings of $e$ and $f$ in $C$ are of the same type.

Suppose now that the ordering of $e$ in $C$ is $xzy$ (or its reverse). Then, the ordering of $e$ is determined by the pair $xy$ and, obviously, $f$ is neither the successor nor the predecessor of $e$ in $C$. If the ordering of $f$ in $C$ is not determined by $xy$, then $x$ and $y$ appear consecutively in the ordering of $f$ in $C$. Moreover, since $xy$ is not a connector, similarly as in the previous case, we obtain that the edge $e$ is the successor or the  predecessor of $f$ in $C$, a contradiction. Thus, the ordering of $f$ in $C$ is determined by the pair $xy$ and, consequently, the orderings of $e$ and $f$ defined by $C$ are of the same type in $H^i$.
\end{inproof}

\begin{claim}
Let $C$ be a tight Euler tour in $H$.  For every $i=1,\ldots,t$, the orderings of all the edges of the subhypergraph $H^i$ defined by $C$ are of the same type, either $T$ or $F$.
\end{claim}

\begin{inproof}{2}
By Claim 1, the orderings defined by $C$ of every two edges in $H^i$ that have two common vertices which do not form a connector are of the same type in $H^i$. It is easy to observe that the pairs of vertices which are connectors in $H^i$ form the set of edges of an acyclic graph. Thus, the orderings defined by the tour $C$ of all edges in $H^i$ are of the same type.

Suppose now the orderings of all edges of $H^i$ are of type $N$. Then, however, the edges of the set $E_1$ in $H^i$ form a tight tour of length $6$. This tour cannot be a part of $C$ because the edges in $E_1$ do not contain any connector. Hence, the orderings of all edges of $H^i$ are of type $T$ or all are of type $F$.
\end{inproof}

\begin{claim}
Let $C$ be a tight Euler tour in $H$ and let $w_{k-1}w_kw_{k+1}w_{k+2}$ be some consecutive terms of $C$ such that $f=w_kw_{k+1}w_{k+2}\in E(H^i)$, for some $i=1,\ldots,t$.

\begin{enumerate}[(i)]
\item

If $e=w_{k-1}w_kw_{k+1}\in E(C_t)$, then the ordering $w_kw_{k+1}w_{k+2}\in E(H^i)$ of $f$ defined by $C$ is of type $T$ or $F$.

\item

If $e=w_{k-1}w_kw_{k+1}\in E(C_6^j)$, for some $j=1,\ldots,p$ and the ordering $w_kw_{k+1}w_{k+2}\in E(H^i)$ of $f$ defined by $C$ is of type $T$ (resp. of type $F$), then the variable $x_i$ appears positively (resp. negatively) in the clause $d_j$.
\end{enumerate}
\end{claim}

\begin{inproof}{3}
\begin{enumerate}[(i)]
\item

Clearly, in this case the pair of vertices $w_kw_{k+1}$ is the connector $v_{\ell_i-1}^av_{\ell_i-1}^b$. Thus, the vertices $v_{\ell_i-1}^a,v_{\ell_i-1}^b$ appear consecutively in the ordering of $f$ in $C$. On the other hand in the ordering of type $N$ of $f$ they do not appear consecutively. Hence, the ordering of $f$ defined by $C$ is of type $T$ or $F$.

\item

In this case the pair of vertices $w_kw_{k+1}$ is a connector equal to $v_{2o_i}^av_{2o_i+1}^a$ or $v_{2o_i}^bv_{2o_i+1}^a$. However, if the ordering of $f$ defined by $C$ is of type $T$, then it is determined by a pair of vertices joined by a horizontal line in Figure \ref{najpiekniejszy}. Thus, the pair $v_{2o_i}^av_{2o_i+1}^a$ (which is joined by a horizontal line in Figure \ref{najpiekniejszy}) cannot be this connector because the vertices of this pair appear consecutively in the ordering of $f$ defined by $C$. Consequently, the connector $w_kw_{k+1}$ is equal to $v_{2o_i}^bv_{2o_i+1}^a$. By the definition of the hypergraph $H$, the variable $x_i$ appears in the clause $d_j$ positively.

The version of Claim 3(ii) concerning the ordering of $f$ of type $F$ can be shown analogously.
\end{enumerate}
\end{inproof}


\medskip
We are ready to prove that the hypergraph $H$ has a tight Euler tour if and only if the formula $c$ is satisfiable.

Let us assume first that $H$ has a tight Euler tour $C$. By Claim 2, for every $i=1,\ldots,t$, the orderings defined by $C$ of all edges of $H^i$ are of the same type -- either $T$ or $F$.

We define a truth assignment $w:X\rightarrow\{{\bf true},{\bf false} \}$ by setting $w(x_i)={\bf true}$ if the orderings of the edges of $H^i$ in $C$ are of type $T$ and $w(x_i)={\bf false}$ otherwise. We shall prove that each clause $d_j$, for $j=1,\ldots,p$, is satisfied under $w$. As $C$ is a tight Euler tour in $H$, there is a subsequence $w_{k-1}w_kw_{k+1}w_{k+2}$ of $C$ such that $e=w_{k-1}w_kw_{k+1}\in E(C_6^j) $ and $f=w_kw_{k+1}w_{k+2}\in E(H^i)$ for some $i$. By Claim 3(ii), if the ordering of $f$ defined by $C$ is of type $T$ (resp. of type $F$), then the variable $x_i$ appears in the clause $d_j$ positively (resp. negatively). In  both cases the truth assignment $w$ evaluates the clause $d_j$ to ${\bf true}$.


Let us assume now that the formula $c$ is satisfiable and let $w$ be a truth assignment for which $c$ evaluates to ${\bf true}$. We shall construct a tight Euler tour in $H$.

Let $(w_0,w_1,w_2,\ldots,w_q)$ and $(u_0,u_1,u_2,\ldots,u_r)$ be two edge-disjoint tight tours, such that they have two common vertices, say $x$ and $y$, that appear consecutively in both tours, say $x=w_0=u_0$ and $y=w_1=u_1$. Then the tight tours can be {\it glued together} to form a tight tour $(x,y,w_2\ldots,w_q,x,y,u_2\ldots,u_r)$ that contains all the edges of the two original tours.

We denote by $C^i$ the tight Euler tour in $H^i$ which is of type $T$ if $w(x_i)={\bf true}$ or of type $F$ if $w(x_i)={\bf false}$. Clearly, the tours $C^1,\ldots,C^t,C^1_6,\ldots,C^p_6,C_t$ are edge-disjoint and contain all edges of the hypergraph $H$. We shall construct a tight Euler tour in $H$ by gluing together these tours.

For every clause $d_j$ there is a literal in this clause that evaluates to ${\bf true}$ in the truth assignment $w$. Let $x_i$ be the variable in this literal. By the definition of the hypergraph $H$, if $x_i$ appears in $d_j$ positively (resp. negatively), then some two consecutive vertices of $C_6^j$ are identified with the vertices $v_{2o_i}^b$ and $v_{2o_i+1}^a$ (resp. $v_{2o_i}^a$ and $v_{2o_i+1}^a$) of $H^i$. Observe that these two vertices appear consecutively in the tour $C^i$ as well because the tour is of type $T$ (resp. of type $F$). Thus, in both cases we can glue the tours $C^i$ and $C^j_6$ together.

We repeat this gluing operation for every tour $C^j_6$ and create $t$ edge-disjoint tight tours, say $T_1,\ldots,T_t$, such that each tour $T_i$ contains all the edges of the hypergraph $H^i$ and the edges of every tour $C^j_6$ are contained in some of the tours $T_i$. To complete the construction of a tight Euler tour in $H$ we observe that each tour $T_i$ can be glued together with the tour $C_t$. Indeed, by the construction of the hypergraph $H$, some two consecutive vertices of $C_t$ are identified with the vertices $v_{\ell_i-1}^a$ and $v_{\ell_i-1}^b$ of $H^i$. Now, these two vertices appear consecutively in the Euler tours in $H^i$ of both types $T$ and $F$. Thus, we can glue the tours $C_t$ and $T_i$ together. We can repeat this gluing operation for every tour  $T_i$ which completes the construction of a tight Euler tour in $H$.

Note that the size of the hypergraph $H$ is linear in $p+t$, as the following calculations show

\begin{align*}
|E(H)| = & \sum_{i=1}^t |E(H^i)|+p|E(C_6)|+|E(C_t)| \\
= & 6\sum_{i=1}^t \ell_i +6p+2t\leqslant 6\sum_{i=1}^t (2\occ(x_i)+4)+6p+2t\\
= & 42p+26t.
\end{align*}

Thus, any algorithm deciding the existence of a tight Euler tour in a 3-uniform hypergraph in time $2^{o(|E(H)|)}$ could be easily transformed (in polynomial time) to an algorithm deciding the satisfiability of an arbitrary \textsc{3-Sat} instance in time $2^{o(p+t)}$, which would in turn contradict the ETH.
\end{proof}

The construction presented above can be easily adapted to show the hardness of determining the existence of a tight Euler trail in a given 3-uniform hypergraph.

\begin{theorem}
\label{ciasnasciezkaEulera}

The problem of determining if a given $3$-uniform hypergraph $H$ has a tight Euler trail is {\em NP}-complete. Moreover, it cannot be solved in time $2^{o(|E(H)|)}$, unless the ETH fails.
\end{theorem}

The proof of Theorem \ref{ciasnasciezkaEulera} is analogous to the proof of Theorem \ref{wykreconetorusy}. The only difference is in connecting the hypergraphs $H^i$, $i=1,\ldots,t$. In the proof of Theorem \ref{wykreconetorusy} we used the cycle $C_t$. Here we ought to use rather a path hypergraph $P_t$ which is defined by $V(P_t)=\{0,\ldots,2t-1\}$, $E(P_t)=\{i(i+1)(i+2):i=0,\ldots,2t-3\}$.

\section{Exact algorithm}
\label{algo}

In this section we describe an asymptotically tight exact algorithm for determining if a given $k$-uniform hypergraph has a tight Euler trail or a tight Euler tour. In fact, we compute the number of all tight Euler trails or tours in the input hypergraph. Our approach is based on a classical algorithm for counting Hamiltonian cycles in graphs (\cite{Karp}).

Let us start with introducing some notation. Let $H=(V,E)$ be a $k$-uniform hypergraph.
For a subset $E' \subseteq E$, by $E'_{k-1}$ we denote the set of all $(k-1)$-element sequences $\bx = x_1x_2\ldots x_{k-1}$ of distinct vertices, such that there exists $e \in E'$ with $\{x_1,x_2,\ldots,x_{k-1}\} \subseteq e$.
To simplify the notation, we will often identify a sequence $\bx$ with the set of its elements, e.g., we will write $\bx \subseteq e$. Also, for $x_0 \in V$, by $x_0\bx$ we mean the sequence $x_0x_1x_2\ldots x_{k-1}$.

For $\by, \bx \in E_{k-1}$, we say that a tight walk (respectively trail) is a tight {\em $\by$-$\bx$-walk} (resp., {\em $\by$-$\bx$-trail}) if it starts with the prefix $\by$ and ends with the suffix $\bx$ (note that the prefix and the suffix may overlap, in particular, we treat $\bx$ as a tight $\bx$-$\bx$-walk of length 0). Clearly every tight Euler trail in $H$ is a tight Euler $\bx$-$\by$-trail for some $\bx,\by \in E_{k-1}$.
Moreover, a tight Euler tour is a tight Euler $\bx$-$\bx$-trail for some $\bx \in E_{k-1}$.

The following version of the well-known inclusion-exclusion principle is the main tool that we use (see e.g. \cite[Theorem 4.2]{FK-Exact}).

\begin{theorem}[Inclusion-exclusion principle, intersection version]\label{inc-ex}
Let $N$ be a finite set and let $\{ Q_i \colon i \in I\}$ be a family of its subsets.
Then
\[
\left |\bigcap_{i \in I} Q_i \right | = \sum_{W \subseteq I} (-1)^{|W|}  \left | N \setminus \bigcup_{w \in W} Q_w \right |.
\]
\end{theorem}

Now we are ready to state the main result of this section.

\begin{theorem}
For every fixed $k \geqslant 3$, the number of (a) tight Euler trails, (b) tight Euler tours in a $k$-uniform hypergraph with $m$ edges can be computed in time $2^m \cdot m^{O(1)}$, using polynomial space.
\end{theorem}

\begin{proof}
Fix $\by = y_1y_2\ldots y_{k-1} \in E_{k-1}$ and $\bx = x_1x_2\ldots x_{k-1} \in E_{k-1}$. So, following the notation used in Theorem \ref{inc-ex}, define $N$ to be the set of all tight $\by$-$\bx$-walks of length $m$ in a $k$-uniform hypergraph $H$ with $m$ edges. Moreover, let $I=E$ be the set of edges of $H$. For an edge $e$, let $Q_e$ denote the set of tight  walks in $N$, which contain the edge $e$. Note that for $W \subseteq E$, $N \setminus \bigcup_{w \in W} Q_w$ is the set of tight $\by$-$\bx$-walks of length $m$ in $H' = (V,E \setminus W)$.
Also, if $W=E$, then $\bigcap_{e \in W} Q_e$ is the set of all tight Euler $\by$-$\bx$-trails in $H$.

Suppose that for every $W \subseteq E$ we can compute $|N \setminus \bigcup_{w \in W} Q_w|$ in time $m^{O(1)}$. Then, by Theorem \ref{inc-ex}, we can compute $|\bigcap_{e \in E} Q_e|$ in time $2^m \cdot m^{O(1)}$.

So, let us focus on computing $|N \setminus \bigcup_{w \in W} Q_w|$.
Let $W$ be a fixed subset of $E$ and let $H':=(V,E')$, where $E' = E \setminus W$. For $d \in \{0,1,\ldots,|E'|\}$ and $\bx' = x_1'x_2'\ldots x_{k-1}' \in E'_{k-1}$, by $T[d,\bx']$ we denote the number of tight $\bx$-$\bx'$-walks of length $d$ in $H'$. Note that if $\bx \notin E'_{k-1}$, then $T[d,\bx']=0$ for all $d$ and $\bx'$.

It is straightforward to verify that $T[d,\bx']$ is given by the following recursion:

\begin{align*}
T[0,\bx'] = & \begin{cases}
1 & \text{ if } \bx' = \bx \in E'_{k-1},\\
0 & \text{otherwise,}
\end{cases}\\
T[d,\bx'] = & \sum_{ \substack{ x_0' \in V, s.t. \\ x_0'\bx' \in E'}} T[d-1,x_0'x_1'\ldots x_{k-2}'] \text{ for } d \geqslant 1.
\end{align*}

Applying this recursive formula, we can use dynamic programming to compute all values $T[d,\bx']$ for $d \in \{0,1,\ldots,|E'|\}$ and $\bx' \in E'_{k-1}$. The size of the dynamic programming table is $O(m \cdot m \cdot k!) = O(m^2)$ and computing each entry takes $O(m \cdot k)=O(m) $ time. So, all values of $T[d,\bx']$ can be computed in time $O(m^3)$ and space $O(m^2)$. Thus, $|\bigcap_{i \in [m]} Q_i|$ can be computed in time $2^m \cdot m^{O(1)}$.

To compute the number of tight Euler trails in the input hypergraph $H$, we need to call the algorithm described above for all $\by,\bx \in E_{k-1}$, i.e., at most $(m \cdot k!)^2 = O(m^2)$ times, and then return the sum of obtained numbers. It is clear that each tight Euler trail will be counted exactly once.

Recall that a tight  Euler tour in $H$ corresponds to exactly $2m$ closed tight Euler trails (i.e.,  $\bx$-$\bx$ tight Euler trails for some $\bx \in E_{k-1}$). Thus to compute the number of tight Euler tours, we need to call the algorithm for each $\bx \in E_{k-1}$ (at most  $m \cdot k! = O(m)$ times), and return the sum of obtained numbers, divided by $2m$.

\end{proof}

\section{Final remarks and open problems}

In this paper we proved that the problem of deciding if a given 3-uniform hypergraph has a tight Euler tour is NP-complete. We believe the same problem for $k$-uniform hypergraphs with $k>3$ is NP-complete too, however, we do not have a proof of this statement. Our proof for $k=3$ does not seem to generalize easily to larger values of $k$.

We think the codegree (defined as the number of edges containing a $(k-1)$-tuple of vertices) is a more natural parameter than the degree in problems concerning tight cycles or tours in $k$-uniform hypergraphs. In fact this is the parameter that appears in the Dirac type theorem on existence of a tight Hamilton cycle in a 3-uniform hypergraph (see \cite{RRS2}). We observed (Theorem \ref{wykreconetorusy}) that the problem of deciding if a 3-uniform hypergraph has a tight Euler tour is NP-complete even for hypergraphs of maximum codegree bounded by $4$. It is not hard to see that the problem becomes polynomial for $3$-uniform hypergraphs with maximum codegree bounded by $2$. It follows from an observation that any tour in such hypergraphs is uniquely determined by any two of its consecutive edges. It is an open problem what the complexity status of our problem is for $3$-uniform hypergraphs of maximum codegree bounded by $3$.


\end{document}